\documentclass{article}
\usepackage{amsmath,amsfonts, amsthm,amssymb,url,graphicx, lineno}
\usepackage[caption=false,font=normalsize,labelfont=sf,textfont=sf]{subfig}
\usepackage{dsfont}
\theoremstyle{plain}
\newtheorem{thm}{\protect\theoremname}
\theoremstyle{plain}

\theoremstyle{plain}

\theoremstyle{plain}

\theoremstyle{plain}
\newtheorem{lem}[thm]{\protect\lemmaname}
\theoremstyle{plain}

\theoremstyle{plain}
\providecommand{\algorithmname}{Algorithm}
\providecommand{\assumptionname}{Assumption}
\providecommand{\lemmaname}{Lemma}
\providecommand{\theoremname}{Theorem}
\providecommand{\corname}{Corollary}
\providecommand{\propositionname}{Proposition}
\providecommand{\remarkname}{Remark}

\newcommand{\cC}{\mathcal{C}}

\newcommand{\cF}{\mathcal{F}}

\newcommand{\cS}{\mathcal{S}}
\newcommand{\cT}{\mathcal{T}}

\newcommand{\cZ}{\mathcal{Z}}
\newcommand{\EE}{\mathbb{E}}
\newcommand{\NN}{\mathbb{N}}
\newcommand{\PP}{\mathbb{P}}

\newcommand{\indic}{{\bf 1}}

\title{Exponential Backoff: Meta-Stability and Implicit Admission Control}
\author{Richard Combes \thanks{Universit\'e Paris-Saclay, CNRS, CentraleSup\'elec, Laboratoire des signaux et syst\`emes, France} \,and Fabien Mathieu \thanks{Swapcard, France} \,and Thomas Bonald \thanks{Institut Polytechnique de Paris, France}}
\begin{document}
\maketitle
\begin{abstract}
We analyze exponential backoff, an algorithm used to share a single communication channel in a distributed manner between several users, in a similar way as many networking standards such as 802.11. We demonstrate that this algorithm has a meta-stable behavior, in the sense that the system oscillates over long time-scales between meta-stable configurations where only a subset of the sources effectively access the channel. In other words, meta-stability creates a form of implicit admission control. We provide both theoretical tools and numerical experiments to understand this phenomenon further.
Our experiments are fully reproducible and the code is publicly available as a Python package (\footnote{available at \url{https://pypi.org/project/slotted-aloha-simulator/}}). 
\end{abstract}

\section{Introduction}
Contention algorithms with exponential backoff are ubiquitous in communication networks, due to their simplicity and distributed nature, and are used for instance in 802.11 protocols. Most existing analyses of exponential backoff rely on a type of mean field approximation known as Bianchi's model \cite{bianchi1996performance}. This approximation, which makes the analysis tractable, is to assume that the sources are decoupled and experience the same collision rate. In this paper we demonstrate that, in general, the predictions made by these analyses are inaccurate. More precisely, we show that exponential backoff has a \emph{meta-stable behavior}, where the system oscillates over long time-scales between meta-stable configurations where, whenever the number of sources is large, only a subset of the sources actually use the channel, the others being idle due to many successive collisions. In other words, the protocol gives rise to a form of \emph{implicit admission control} where only a subset of the sources actually access the channel. This phenomenon of meta-stability was already observed in \cite{vvedenskaya2007multi} but under the fundamentally incorrect assumption that the parameters of the protocol depend on the number of sources. Indeed, to the best of our knowledge, all practical implementations of exponential backoff impose \emph{fixed parameters} that do not depend on the number of sources. Indeed, the very goal of protocols such as exponential backoff is to adapt to the number of sources without prior information. Due to practical design constraints, sources must adapt solely on previous transmission successes and failures, without any additional information to guide their behavior.

 Our main contribution is to prove that meta-stability occurs in the true setting where the parameters of the protocol are fixed, independently of the number of sources. Furthermore, and perhaps counter-intuitively, we claim that meta-stability is in fact \emph{desirable} and creates a form of \emph{implicit admission control}, allowing an efficient utilization of the channel. Indeed, most work on design of network protocols focuses on stable algorithms, and regards meta-stable and/or unstable algorithms as undesirable. We provide theoretical tools to understand when meta-stability may occur and we complement those findings with extensive, long time-scale numerical experiments confirming that indeed, meta-stable behavior appears and that exponential backoff genuinely performs implicit admission control without any prior knowledge of the number of sources.

The rest of the article is organized as follows:
Section~\ref{sec:rw} reviews the related work. Section~\ref{section:model} introduces the proposed model, for which Section~\ref{section:meanfield} provides a mean-field approximation. Theoretical results for two or more competing sources are presented in Section~\ref{sec:n2} and~\ref{sec:n} respectively. Section~\ref{sec:simu} showcases our numerical results and Section~\ref{sec:conclu} concludes.

\section{Related work}\label{sec:rw}

Due to the ubiquity of exponential backoff algorithms, their analysis has a rich history starting with the seminal works of Bianchi \cite{bianchi1996performance,bianchi,bianchi2}. Further refinements were proposed by~\cite{kwak} and~\cite{kumar2005new}. All of those works are based on an approximation we refer to as \emph{mean field approximation} where all sources are assumed to be decoupled and seen as i.i.d. copies of a single source experiencing the same collision probability.
We give a brief presentation of this approximation in section~\ref{section:meanfield}, and while it makes the system tractable, it hides the meta-stable behavior which we highlight in the present work.

While we focus on the saturated setting where sources always have packets to transmit, Tobagi and Kleinrock considered non-saturated sources, albeit without backoff \cite{tobagi}. Stability of the exponential backoff algorithm in a dynamical system with arrivals and departures of sources was considered by Aldous~\cite{aldous1987ultimate}, in contrast with our setting where the number of sources is fixed. Several works proposed Markov chain models for the case of non-saturated sources, often using a similar decoupling approximation as that of Bianchi \cite{garetto2005performance, malone2007modeling,dao2008new}. Those results were used to optimize 802.11 protocols \cite{6644966} and flow control \cite{baiocchi2023flow}. Our work sheds a new light on some of the conclusions from these works since the decoupling approximation might not reflect the actual system performance. 
Simpler approaches for non-saturated sources were explored in \cite{daneshgaran2007linear,zhao2008simple,ghaboosi2008modeling,tickoo2008modeling,zhao2009simple,vijayasankar2010analytical}, by neglecting the interactions between sources created by exponential backoff.
An analysis of the  Markov chain of saturated sources was proposed in \cite{sharma2009performance}, in order to justify the decoupling assumption.  Our work shows that in some cases, there are qualitative differences due to meta-stability.

Meta-stability was first explored in \cite{vvedenskaya2007multi}, and generally considered as an undesirable phenomenon~\cite{ashrafi2016random}. However, these works consider scaling regimes where the probability of transmission is a decreasing function of the number of sources. We consider a more realistic model where the number of sources is unknown, since the role of the algorithm is to adapt the transmission policy of each source to the number of sources.  In fact, we claim that meta-stability can be desirable property, as it introduces a form of implicit admission control when the number of sources is too large. 
Many more recent works explore the idea of mean-field approximation for saturated sources \cite{proutiere2005random,bordenave2012asymptotic,ashrafi2016random} and non-saturated sources \cite{cecchi2016csma,michalopoulou2016mean}. In all these works, the transmission probability is inversely proportional to the number of sources. Again, we consider the more realistic scenario where the transmission probability is independent of the number of sources. This gives rise to meta-stability, a key phenomenon that is absent from these analyses due to the particular scaling regime considered therein. 

\section{Model}

\label{section:model}

\subsection{Exponential Backoff}
We consider $N$ identical sources sharing a common channel in a distributed manner. Time is slotted, each packet takes one time slot to transmit, in the full-buffer setting where sources always have packets to transmit. If a single source attempts to transmit in a given time slot, this transmission is successful, whereas if several sources attempt to transmit in the same time slot, a collision occurs and none of the packets are transmitted. We consider a protocol  called \emph{exponential backoff}: at each time slot, each source attempts to transmit with probability $p_0 \alpha^{x}$, where $x$ is the number of previous collisions experienced since the last successful transmission and $p_0,\alpha\in  (0,1)$ are two protocol parameters: $p_0$ is the initial transmission probability (for the first packet, or after any successful transmission) and the transmission probability decreases by a factor $\alpha$, called the backoff rate, upon each collision. All sources use the same protocol. This situation is a stylized model for exponential backoff protocols such as those used in CSMA (see discussion below). 

\subsection{Markov Chain}
At time $t \in [T] = \{1,...,T\}$, with $T$ is the time horizon, the state of source $i\in [N]$ is denoted by $X_i(t) \in \NN$. This is the number of collisions experienced from her last successful transmission up to time $t-1$. Sources draw $Z_1(t),...,Z_N(t)$ independent Bernoulli variables with respective means $p_0 \alpha^{X_1(t)},...,p_0 \alpha^{X_N(t)}$ and source $i$ transmits if and only if $Z_i(t) = 1$. Source $i$ will successfully transmit if and only if $Z_i(t) \prod_{j \ne i}(1- Z_j(t)) = 1$. The state of the system $X(t) = (X_1(t),...,X_N(t))$ is clearly a discrete time Markov chain over $\NN^N$. For any $x, x'$ in $\NN^N$, we denote by
\begin{align}
	P(x,x') = \PP( X(t+1) = x'| X(t) = x)
\end{align}
the transition probability of the Markov chain from state $x$ to state $x'$. Let $e_i$ be the vector with 1 in the $i$th coordinate and 0 otherwise. We have:
\begin{align}
P(x,x') = p_0 \alpha^{x_i} \prod_{j \ne i} (1-p_0\alpha^{x_j}) 
\end{align}
if $x_i>0$ and $x' = x - x_ie_i$ (source $i$ with positive state successfully transmits and resets its state),
\begin{align}
P(x,x') = \Big(\prod_{i \in S} p_0 \alpha^{x_i} \Big) \Big( \prod_{j \not\in S} (1-p_0 \alpha^{x_j}) \Big)
\end{align}
if $x' = x + \sum_{i\in S}e_i$ for some subset  $S \subset [N]$ such that $|S|>1$ (sources in $S$ collide and increment their state), and 
\begin{align}
P(x,x') = \prod_{j \in [N]} (1-p_0\alpha^{x_j})+p_0\sum_{i: x_i=0}\prod_{j\neq i}(1-p_0\alpha^{x_j})
\end{align}
if $x'=x$ (no source transmits or a source with null state transmits). In all cases not listed above $P(x, x') = 0$. 

In Figure~\ref{fig:n_2_state_diagram} we represent the state space of $X(t)$ for $N=2$ sources, where the colors of the arrows represent the type of transition: both sources collide (red); both sources do not transmit (black); source $1$ transmits successfully (green); source $2$ transmits successfully (blue).
\begin{figure}[!ht]
\centering
\includegraphics[width=0.7\textwidth]{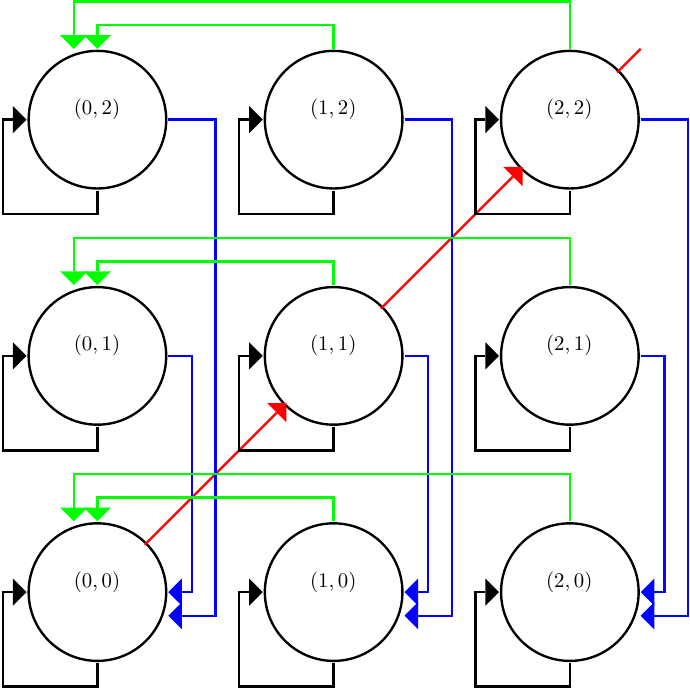}
\caption{State space for $N=2$ sources}
\label{fig:n_2_state_diagram}
\end{figure}



\subsection{Long-Term Behavior}
The Markov chain  $X(t)$ is irreducible. To predict its long term behavior, we must know if it is \emph{transient} (i.e. some states are visited a finite number of times), \emph{positive recurrent} (i.e. all states are visited infinitely many times and the expected time between two visits to a given state is finite) or \emph{null recurrent} (i.e. all states are visited infinitely many times and the expected time between two visits to a given state is infinite). An equivalent term for positive recurrence used in the queuing literature is \emph{stability}. If the process is positive recurrent, then its long time behavior will be given by its stationary distribution, which is the unique solution to the balance equations:
\begin{equation}
    \pi(x) = \sum_{x' \in \NN^N} \pi(x') P(x',x) \text{ for all } x \in \NN^N.
    \label{eq:balance}
\end{equation}
Now, to the best of our knowledge, it is not possible to solve the balance equations in closed form, even in the simplest case  $N=2$. The solution can be approximated numerically by truncating the state space and solving a linear system to perform numerical experiments. Furthermore, more fundamentally, there are some regimes in which no solution exists because the Markov chain is not positive recurrent, as our work will demonstrate. In those regimes, other strategies are needed in order to predict performance. 
\subsection{Exponential Backoff vs Uniform Backoff}
In our work, after experimenting $x$ collisions, a source will wait for a geometrically distributed time with parameter ${p_0 \alpha^x}$ whereas in classical protocols like CSMA this waiting time is uniformly distributed over some contention window. This stylized model simplifies the state space of the Markov chain (compare for example Figure~\ref{fig:n_2_state_diagram} to~\cite[Fig. 4]{bianchi2}) and eases the theoretical analysis.  In Section~\ref{sec:validation}, we show through numerical experiments that  both systems behave similarly, so that our theoretical predictions are relevant and are not an artifact of our model.

\section{Mean-Field Approximation}\label{section:meanfield}

We first consider the standard approach based on mean-field approximation in our setting with exponential backoff. It will be shown in the following sections that the conclusions drawn from this approach can be erroneous due to the phenomenon of meta-stability.

\subsection{Mean-Field Equations}
When the Markov chain is positive recurrent, the stationary distribution exists and is the unique solution to the balance equations \eqref{eq:balance}. Since these equations cannot be solved in general, a popular approach introduced in the seminal work of Bianchi \cite{bianchi} is to use a decoupling assumption, a form of mean-field approximation used in physics: we then consider a fictitious system where the states of the sources are i.i.d. This is an approximation because in the real system, the states of the sources are correlated due to the collisions. It is expected that, for  large values of $N$, the states of the sources become independent. This is called the {mean-field} approximation in analogy with physics where, in large systems of particles, it is frequently assumed that particles act as i.i.d. copies of a single particle, each experiencing the same {\it mean field}.

Formally, the mean-field approximation consists in  considering an alternate process $\tilde{X}(t) = (\tilde{X}_1(t),...,\tilde{X}_N(t))$ where $\tilde{X}_1(t),...,\tilde{X}_N(t)$ are i.i.d. and whenever a source attempts to transmit, this transmission results in a collision  with some fixed probability $q$ to be determined. From independence, the transition probabilities of this process are given by:
\begin{align}
	\tilde{P}(x,x') 
	&= \PP( \tilde{X}(t+1) = x'| \tilde{X}(t) = x) \\
	&= \prod_{i \in [N]} \tilde{p}(x_i,x_i')
\end{align} 
with
\begin{align}
\tilde{p}(x_i,x_i') = \PP( \tilde{X}_i(t+1) = x_i'| \tilde{X}_i(t) = x_i)
\end{align}
given by
\begin{align}
\label{eq:mf_transition}    
\tilde{p}(x_i,x_i') = \begin{cases} 
(1-q) p_0 \alpha^{x_i} & \text{ if } x_i>0, x'_i=0,  \hfill (i) \\
q p_0 \alpha^{x_i} & \text{ if } x'_i=x_i+1, \hfill (ii) \\ 
1-p_0 \alpha^{x_i} & \text{ if } x_i>0, x'_i=x_i, \hfill (iii) \\ 
1-q p_0 & \text{ if } x_i=x'_i=0, \hfill (iv) \\ 
0  & \text{ otherwise,} \hfill  (v) \end{cases} 
\end{align}
where the cases enumerated above correspond to: $(i)$ source $i$ successfully transmitting; $(ii)$  source $i$ attempting to transmit and colliding; $(iii)$ source $i$ not transmitting; $(iv)$ the edge case $\tilde{p}(0, 0)$ that combines cases $(i)$ and $(iii)$. The mean field approximation therefore consists in approximating $X(t)$ by $\tilde{X}(t)$, and while in some cases this approximation is reasonable, we will demonstrate that there are also cases in which the two processes behave in a completely distinct manner.
\subsection{Stationary Distribution and Fixed-Point Equation}
\label{sec:mean_field_geo}
The main advantage of the mean field approximation is that the balance equations
\begin{equation}
\tilde{\pi}(x) = \sum_{x' \in \NN^N} \tilde{\pi}(x') \tilde{P}(x, x') \text{ for all } x \in \NN^N
\end{equation}
can be solved in closed form, and the solution is
\begin{equation}
\label{eq:mf_geometric}
\tilde{\pi}(x) = \prod_{i \in [N]} \tilde{\pi}_i(x_i), \text{ with } \tilde{\pi}_i(x_i +1) = \frac q \alpha \tilde{\pi}_i(x_i).
\end{equation}
In particular, a stationary distribution exists if and only if $q<\alpha$, in which case $\tilde{X}_1(t),...,\tilde{X}_N(t)$ are i.i.d.  Geometric($q / \alpha$). In stationary regime, source $i$ transmits with probability
\begin{align}
\PP(\tilde{Z}_i(t) = 1) &= \sum_{x_i \in \NN} \PP(\tilde{Z}_i(t) = 1|\tilde{X}(t) = x_i) \PP(\tilde{X}(t) = x_i) \\
&= \sum_{x_i \in \NN} p_0 \alpha^x (q/\alpha)^{x_i} (1-q/\alpha) \\
&= p_0 \frac{ 1 - q/\alpha}{1 - q}
\end{align}
Therefore $q$ must obey the fixed point equation:
\begin{align}
q &= \PP\left( \prod_{j \ne i} (1-\tilde{Z}_j(t)) = 1\right) \\
  &= 1-[1-\PP\left( \tilde{Z}_i(t) = 1\right)]^{N-1}\\ 
 &= 1 - \left[1 - p_0 \frac{ 1 - q/\alpha}{1 - q} \right]^{N-1} \\
 &= f(q)
\label{eq:fixed_point}
\end{align}
This fixed-point equation always has a unique solution over $(0,\alpha)$ since $q \mapsto f(q)$ is continuous non-decreasing, $f(0)> 0$ and $f(\alpha)< \alpha$. 

 \subsection{Performance predictions}
 Let $q^\star$ the unique solution to the fixed-point equation above, which can be computed rapidly using Newton's method.
 The individual throughput of source $i$ in stationary state is
 \begin{align}
	 \lim_{T \to \infty} R_i(T) &= \PP(\tilde{Z}_i(t) \prod_{j \ne i} (1-\tilde{Z}_j(t)) ) \\
 &=  \PP(\tilde{Z}_i(t) = 1)(1-\PP(\tilde{Z}_i(t) = 1))^{N-1} \\
 &= p_0 ( 1 - q^\star/\alpha)
 \end{align}
 since $\tilde{Z}_1(t),...,\tilde{Z}_N(t)$ are i.i.d and the total goodput is
 \begin{equation} 
 \lim_{T \to \infty} \overline{R}(T) = \sum_{i \in [N]} R_i(T) =  N p_0 ( 1 - q^\star/\alpha)
 \end{equation}
 The mean-field approximation therefore predicts that slotted Aloha is  perfectly fair since all sources $i$ achieve the same goodput $R_i(T) = (1/N) \overline{R}(T)$ when $T \to \infty$. As we shall see, while in some cases the mean-field approximation captures the actual performance, sometimes its predictions are incorrect for both total goodput and  fairness, because of meta-stability.

\section{Long-Term Behavior: The Case of Two Sources}\label{sec:n2}

We now investigate the long-term behavior of the system, starting with the simple case of $N=2$ sources.

\subsection{Positive Recurrent and Null Recurrent Regimes}

  Theorem \ref{thm:twosources} shows that there exists two main regimes, depending on the value of backoff rate $\alpha$: a stable (or positive recurrent) regime when $\alpha > p_0$, and a meta-stable (or null recurrent) regime when  $\alpha \le  p_0$. Interestingly, the Markov chain is never transient: all states occur infinitely many times. 

\begin{thm}\label{thm:twosources}
	Consider $N=2$ sources. If $ \alpha > p_0$ then the Markov chain $X(t)$ is positive recurrent and admits a unique stationary distribution, solution to the balance equations $\pi = \pi P$. If $  \alpha \le p_0$ then the Markov chain $X(t)$ is null recurrent and admits no stationary distribution.
\end{thm}

{\bf Proof:} \underline{Recurrence}. We first prove that all states are recurrent, i.e., they are visited infinitely often, for any $p_0$ and $\alpha$.  Let $Y(k) = X(\tau_k)$ for $k \in \NN$ where time $\tau_k = \min \{t > \tau_{k-1}: X(t) \ne X(\tau_{k-1})\}$ is the $k$-th time the value of $X(t)$ changes, with $\tau_0 = 0$. $Y(k)$ is $X(t)$ sampled at those instants where it changes, so if a state that is visited infinitely often by $Y(k)$  is also visited infinitely often by $X(t)$. Since $X(t)$ is a Markov chain with probability transition matrix $P(x,x')$, $Y(k)$ is also a Markov chain, with probability transition matrix
\begin{equation}
Q(x,x') = \frac{P(x,x')}{1 - P(x,x)},\quad x,x'\in \NN^2.
\end{equation}
Let $x_1, x_2\in \NN$, and assume without any loss of generality that $x_1 \le x_2$. We prove that if $Y(k) = (x_1, x_2)$, then there is some probability that $Y(k+2) = (0,0)$, by considering the path $(x_1, x_2) \to (0,x_2) \to (0,0)$. 
The transition probabilities for $Y(k)$ are:
\begin{align}
Q((x_1,x_2),(0,x_2)) 
&=\frac{1-p_0 \alpha^{x_2} }{ 1 + \alpha^{x_2-x_1}-p_0 \alpha^{x_2}}  \\
&\geq \frac{1-p_0}{2}
\end{align}
using the bounds $1-p_0 \alpha^{x_2} \ge 1-p_0$ since $x_2 \ge 0$ and $1 + \alpha^{x_2-x_1}-p_0 \alpha^{x_2} \le 2$ since $x_1 \le x_2$, and
\begin{align}
Q((0,x_2),(0,0)) &= 1-p_0.
\end{align}
Hence for any $(x_1,x_2)\in \NN^2$ and any $k\ge 0$:
\begin{align}
\PP&( Y(k+2) = (0,0) | Y(k) = (x_1,x_2)) \\ 
&\ge Q((x_1,x_2),(0,x_2)) Q((0,x_2),(0,0))   \\
&\ge \frac{(1-p_0)^2}{2}.
\end{align}
Summing over both $(x_1, x_2)$ and $k$ shows that the expected number of visits to state $(0,0)$ is infinite. This state is recurrent for $Y(k)$ and thus also for $X(t)$. Since the Markov chain $X(t)$ is irreducible,  it is recurrent.

\underline{Positive recurrence}. Consider $\alpha> p_0$, we prove that $X(t)$ is positive recurrent by a coupling argument. Let $X^+(t)$ the value of $X(t)$ in an alternate system where, when a source transmits, a collision occurs with probability $p_0$. One may readily check that $X^+(t) \ge X(t)$ almost surely for all $t$. The entries of $X^+(t)$ are i.i.d. so it suffices to prove that $X_i^+(t)$ is positive recurrent and since $X_i^{+}(t)$ is a Markov chain, it is sufficient to prove that $X_i^{+}(t)$ admits a stationary distribution. By inspection, a stationary distribution for $X_i^{+}(t)$ is
\begin{equation}
\pi^{+}(x_i) = (p_0/\alpha)^{x_i} (1-p_0/\alpha)
\end{equation}
since it satisfies the  balance equations:
\begin{align}
\pi^{+}(x_i+1) = \pi^{+}(x_i+1) (1 - p_0 \alpha^{x_i+1}) + \pi^{+}(x_i) p_0^2 \alpha^{x_i}.
\end{align}
Hence $X_i^+(t)$ is positive recurrent, and so is $X(t)$.

\underline{Null recurrence}. We consider the case $ \alpha \le p_0$ and prove that $X(t)$ is null recurrent by contradiction. If $X(t)$ is positive recurrent it admits a stationary distribution $\pi$, with $\pi(x) > 0$ for all states $x$ by irreducibility and satisfying the balance equations $\pi = \pi P$ so that
\begin{align}
\frac{\pi(x_1+1,1)}{\pi(x_1,0)} = \frac{p_0^2 \alpha^{x_1}}{ 1 - ( 1 - p_0\alpha^{x_1+1})(1 - p_0 \alpha)}.
\end{align}
Ignoring transitions $(x_1+1,x_2) \to (x_1+1,0)$ for $x_2 \geq 2$, we obtain the lower bound:
\begin{align}
\pi(x_1+1,0) &\geq \pi(x_1+1,0) (1 - p_0\alpha^{x_1+1}) \\
&+ \pi(x_1+1,1) p_0 \alpha( 1 - p_0 \alpha^{x_1+1}).
\end{align}
Using the previous equation,
\begin{align}
\pi(x_1+1,0) &\geq \pi(x_1+1,1) \alpha^{-x_1}(1 - p_0 \alpha^{x_1+1}) \\ &= \pi(x_1,0) \frac{p_0^2 (1 - p_0 \alpha^{x_1+1}) }{1 - ( 1 - p_0\alpha^{x_1+1})(1 - p_0 \alpha)}.
\end{align}
We deduce that:
\begin{equation}
\frac{\pi(x_1+1,0)}{\pi(x_1,0)} \geq \frac{p_0^2 (1 - p_0 \alpha^{x_1+1})}{1 - ( 1 - p_0\alpha^{x_1+1})(1 - p_0 \alpha)} \sim_{x_1 \to +\infty}  \frac{p_0}{\alpha}.
\end{equation}
If $p_0 > \alpha$ then $\pi(x_1, 0)\to +\infty$, a contradiction and if $p_0 = \alpha$ 
\begin{equation}
\frac{\pi(x_1, 0)}{\pi(0, 0)} \ge \prod_{n\in \NN} \frac{1-\alpha^{n+2}}{1+\alpha^n - \alpha^{n+2}} = c > 0.
\end{equation}
Hence $\sum_{x_1 \in \NN} \pi(x_1,0) = +\infty$, and $\pi$ cannot be a stationary distribution, a contradiction,  so $X(t)$ is null recurrent. 
\subsection{Stationary Distribution in the Positive Recurrent Regime}
In the stable regime  $\alpha > p_0$, a unique stationary distribution $\pi$ exists. While this distribution cannot be computed in closed form, it can be evaluated by truncating the state space $\NN^2$ to a large enough finite subset of $\NN^2$, and then solving the balance equations numerically.

Consider the return time to state $x$, given by
$\tau(x) = 1 / {\pi(x)}.$
This is also the expected time  between two visits to state $x$. 
In Figure \ref{fig:two_sources_return_time}, we present the return time to state $(0,0)$ obtained by numerical evaluation. When $p_0 \to 0$, sources never transmit, the system remains in state  $(0,0)$ at all times and $\tau(0,0) \approx 1$; when $p_0 \to \alpha$ we have $\tau(0,0) \to \infty$: the return time to state $(0,0)$ becomes infinite, and the system tends to be meta-stable.

\begin{figure}[!ht]
\centering
\includegraphics[width=0.7\textwidth]{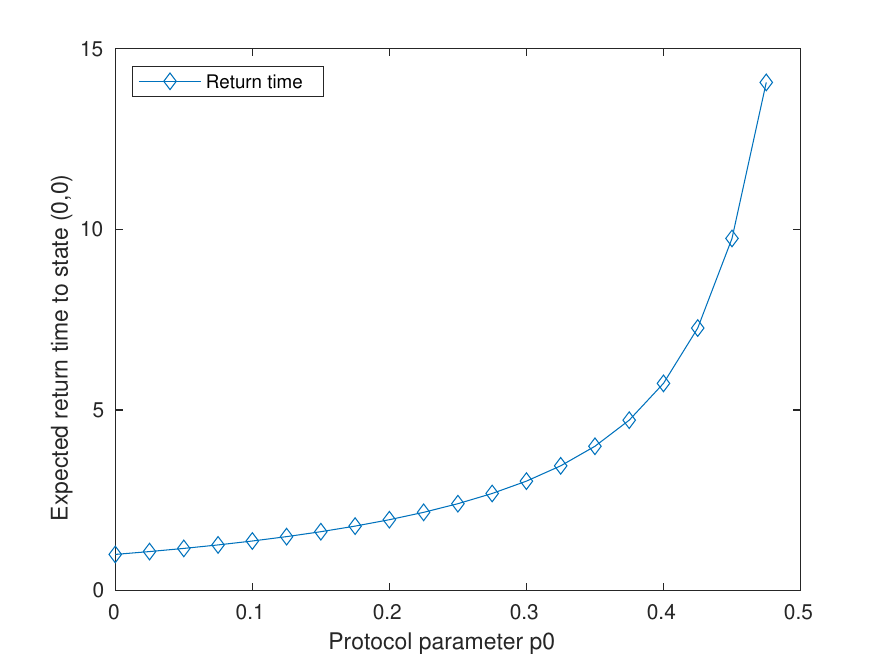}
\vspace{-.2cm}
\caption{Expected return time to state $(0,0)$ with respect to $p_0$ for $N=2$ sources, $\alpha=1/2$}
\label{fig:two_sources_return_time}
\end{figure}

\subsection{Meta-Stability in the Null Recurrent Regime}
In the meta-stable regime  $\alpha \le p_0$, it becomes more difficult to assess which parts of the state space $\NN^2$ are typically explored by the Markov chain $X(t)$. We propose the following interpretation: since the expected return time to state $(0,0)$ is infinite, this means that $X(t)$  spends a  large amount of time in regions of the state space where either the first component or the second component of $X(t)$ is large. For instance, assume that $X_2(t) \gg 1$ during a long time interval. This means that source $2$ will almost completely stop transmitting, and source $1$ will have almost free access to the channel. Hence, in the meta-stable regime, the algorithm performs an implicit \emph{admission control} in which one of the two sources is excluded while the other one is given total control of the channel, without any collisions. The reason why this is a meta-stable equilibrium is because, after a very long time, the  excluded source will eventually be able to transmit and re-integrate the system (since the Markov chain is  recurrent). 

This form of implicit admission control shows why the mean-field approximation, which is based on the assumption that the components of $X(t)$ are i.i.d., cannot accurately predict the system behavior.  The components of $X(t)$ are \emph{strongly correlated} in the meta-stable regime: if $X_1(t)$ is very large then $X_2(t)$ must be close to $0$, and vice-versa. 

In Figure \ref{fig:two_sources_state_time} we present a typical trajectory of $X(t)$ for $N=2$ sources in the meta-stable regime where $p_0=2/3$ and $\alpha=1/2$ during a time interval of $T= 10^5$ time slots. We see that the process alternates between epochs in which one of the sources does not transmit, and the other source behaves as if she were alone in the system. To emphasize this, we draw thick lines to denote long intervals (of length greater than $10^3$ time slots) during one of the  sources does not transmit. Not only does this illustrate meta-stability, but it also suggests that meta-stability can  be \emph{desirable} in the sense that it alternates between state-space regions in which very few collisions  occur: in short, exponential backoff in the meta-stable regime enables to approximate the optimal centralized algorithm (time division multiple access) in a completely decentralized manner. 

\begin{figure}[!ht]
\centering
\includegraphics[width=0.7\textwidth]{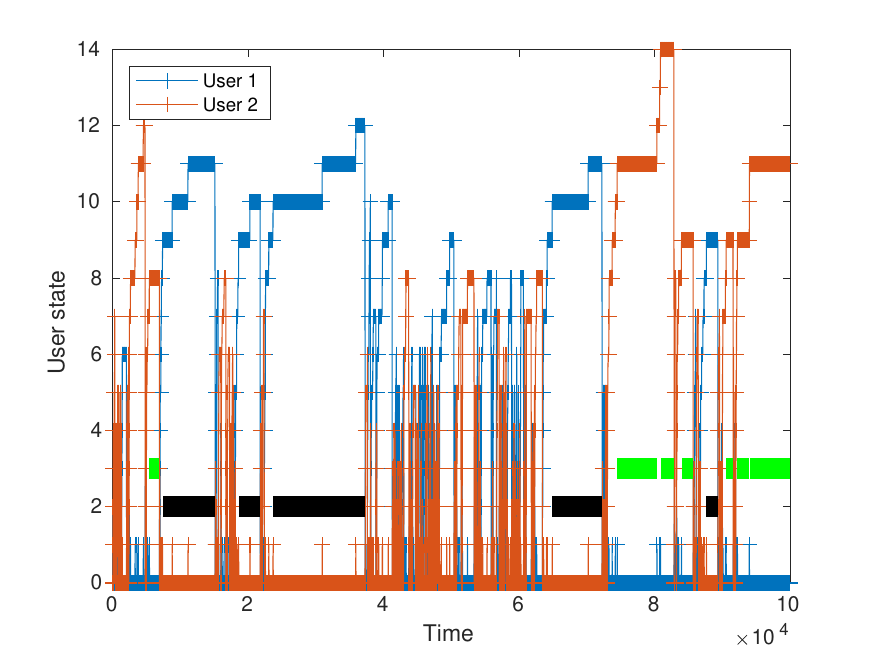}
\caption{Example of trajectory for $N=2$ sources, $p_0=2/3$, $\alpha=1/2$}
\label{fig:two_sources_state_time}
\end{figure}

\section{Long-Term Behavior: General Case}
\label{sec:n}

We now consider the long-term behavior in the  case $N > 2$.

\subsection{Stability Region}

The stability region is characterized by the following result, which is a major contribution of our work. The proof, which is based on  delicate coupling arguments, is detailed in the rest of the section.

\begin{thm}\label{thm:N_sources}
	Consider $N>2$ sources. Let $\psi(n,\alpha)$ be the probability that there exists at least one source attempting to transmit in stationary state, in a system with $n \in [N]$ sources.  
\begin{itemize} 
\item If $\alpha > \max_{n \in [N-1]} \psi(n,\alpha)$ then $X(t)$ is positive recurrent, and hence admits a unique stationary distribution. 
\item If $\alpha < \psi(N-1,\alpha)$ then $X(t)$ is not positive recurrent, no stationary distribution exists, and the expected return time to initial state $(0,...,0)$ is infinite. 
\end{itemize}
In particular,   $X(t)$ is positive recurrent whenever $\alpha > 1 - (1-p_0)^{N-1}$.
\end{thm}

 Theorem~\ref{thm:N_sources}  shows that there are cases in which the Markov chain is positive recurrent, and other cases where the Markov chain  is not positive recurrent. In the latter case, the expected return time to $(0, \ldots, 0)$ is infinite, so that in the long run, a form of implicit admission control occurs and a subset of sources is blocked, the others accessing the channel with few collisions. Observe that this key phenomenon is not captured by  the mean-field approximation.
 
 The result also shows that, by scaling the parameters of the exponential backoff algorithm with $N$, it is always possible to  ``force'' the Markov chain to be positive recurrent, which is in fact the only regime considered in  the literature; the phenomenon of meta-stability, which occurs in the  realistic scenario where the protocol parameters do {\it not} depend on $N$, is precluded. Interestingly,  the long-term behavior depends  on $\psi(n,\alpha)$,  the probability that at least one source transmits in a system with only  $n < N$ sources. In particular, $\psi(N-1,\alpha)$ can be interpreted as the collision probability seen by a  source that is removed from the system and attempts to re-enter it. It is therefore natural to interpret this value as the ``ambient noise'' seen by a source, which is reminiscent (but very different, as the source is removed from the system) of the collision probability $q$ considered in the mean-field approximation.

\subsection{Proof: Positive Recurrence}
Let $\Psi =\max_{n \in [N-1]} \psi(n,\alpha)$. We first consider the case  $\alpha > \Psi$ and prove that $X(t)$ is positive recurrent. We proceed by recursion, and show that if the results holds for at most $N-1$ sources, it also holds for $N$ sources. 

\underline{Step 1: Preliminary}
Since $\Psi < \alpha < 1$, there exists $\epsilon > 0$ such that $\Psi + 2\epsilon < \alpha(1-\epsilon)$. Define $v(x) = x \alpha^{-x}$ and $V(x)= \sum_{i\in[N]} v(x_i)$. Given $k \in [N-1]$ and $y = (y_1,...,y_k)$, define $\phi(y,k)$ the coupling time of process $(X_{1}(t),...,X_{k}(t))$ starting in state $y$, in an alternate system where sources $k+1,...,N$ have been removed. Since $\alpha > \Psi \ge \psi(k,\alpha)$ and we have assumed that the result holds for any $k \le N-1$, $\phi(y,k)$ has finite expectation. For any $B > 0$, define 
$
\Phi(B) = \max_{||y||_\infty \leq B,  k \in [N-1]} \EE(\phi(y,k)) < \infty
$
since we have taken the maximum over a finite set. Define a sequence $(b_1,...,b_{N}) \in \NN^N$ such that, for all $k \in [N-1]$: 
$
	2 (k + 1)p_0 \alpha^{b_k} \Phi\left( \sum_{i=1}^{k-1} b_i \right)  < \epsilon
$
and $1 - (1-p_0 \alpha^{b_k})^{k} \leq \epsilon$. This sequence can be constructed by induction on $k$. Define $\bar{b} = \sum_{i \in [N]} b_i$. Now choose $T \geq 0$ such that for all $k$:
\begin{equation}
\sum_{t \in [T]} (1-p_0 \alpha^{b_k})^{(t-1)(k+1)} \geq \frac{1}{2 (k + 1)p_0 \alpha^{b_k}},
\end{equation}
using the inequality $(1+u)^{k+1} \geq 1 + (k+1) u$. Such a $T$ exists since the sum converges when $T \to \infty$.

\underline{Step 2: Stopping time} Let $x = (x_1,...,x_N)$, with $x_1 \geq \bar{b}$. Without loss of generality assume that $x_2 \leq x_3 \leq ... \leq x_N$. Next define 
$
K = \min \{ i \geq 2: x_{i+1} - x_{i} \geq b_i\}.
$
with $K = N$ if this set is empty. It is noted that $x_2,...,x_K \leq \sum_{i \in [K-1]} b_i$, and that $x_{K+1},...,x_N \geq x_K + b_K \geq b_K$. We partition $\{2,...,N\}$ in two sets $\cS = \{K+1,...,N\}$ and $\cS' = \{2,...,K\}$. Define $\tau_1,\tau_{K+1},...,\tau_{N}$ the first transmission times of sources $1,K+1,...,N$. Define the stopping time  
$
\cT = \min(\tau_1,\tau_{K+1},...,\tau_N,T)
$
Define $\phi$ the coupling time of $(X_2,...,X_K)$ in an alternate system where sources in $\{1\} \cup \cS$ have been removed.

\underline{Step 3: Drift} Let \begin{equation}\Delta_1(x) = \EE( v( X_1(\cT + 1)) | X(0) = x ) - v(x_1)\end{equation} For $A \subset \{2,...,N\}$, define $\cC_A(t)$ the event $\{ \max_{i \in A} Z_i(t) = 1\}$ i.e. a source $i \in A$ transmits at time $t$. Consider the system starting in state $x$ at time $0$. If $\tau_1 \neq \cT$ we have $X_1(\cT + 1) = x_1$. If $\tau_1 = \cT$ and $\cC_{\cS \cup \cS'}(\tau_1)$ occurs, then $X_1(\cT + 1) = x_1 + 1$. Finally, if $\tau_1 = \cT$ and $\cC_{\cS \cup \cS'}(\tau_1)$ does not occur, $X_1(\cT + 1) = 0$. Hence, the drift $\Delta_1(x)$ equals 
\begin{align}
\PP( \tau_1 = \cT ) (v(x_1 + 1) \PP(\cC_{\cS \cup \cS'}(\tau_1) | \tau_1 = \cT) - v(x_1) ). 
\end{align}
\underline{Step 4: Collision probabilities} Times $\tau_1,\tau_{K+1},...,\tau_{N}$ are independent and geometrically distributed with respective parameters $p_0 \alpha^{x_1}, p_0 \alpha^{x_{K+1}}, ..., p_0 \alpha^{x_{N}}$. Define $a = p_0 \alpha^{x_1}$ and $\sigma = 1 - \prod_{i=K+1}^N (1 - p_0 \alpha^{x_i})$ if $K < N$ and $\sigma = 0$ otherwise. We deduce:
\begin{align}
\PP( \tau_1 = \cT ) = a \sum_{t=1}^T ( (1-a)(1-\sigma) )^{t-1}.
\end{align}
Recall that $x_1 \geq \bar{b} \geq b_{K}$ and $x_{K+1},...,x_N \geq b_K$. Hence $1 - a = 1 - p_0 \alpha^{x_1} \geq 1 - p_0 \alpha^{b_K}$ and 
$1 - \sigma =  \prod_{i=K+1}^N (1 - p_0 \alpha^{x_i}) \geq (1-p_0 \alpha^{b_K})^{K}$, 
and by definition of $T$:
\begin{align}
	\PP( \tau_1 = \cT ) &\geq a \sum_{t=1}^{T} (1-p_0 \alpha^{b_K})^{(t-1)(K+1)}  \\
			    &\geq  \frac{a}{2 (K + 1) p_0 \alpha^{b_K}}
\end{align}
We now upper bound the collision probability $\PP(  \cC_{\cS \cup \cS'}(\tau_1) | \tau_1 = \cT)$ by decomposing $\cS \cup \cS'$:
\begin{align}
\PP(  \cC_{\cS}(\tau_1) | \tau_1 = \cT) = \frac{\PP(\tau_1 = \min(\tau_{K+1},...,\tau_{N})) }{\PP( \tau_1 = \cT )} = \sigma 
\end{align}
which is smaller than $1 - (1-p_0 \alpha^{b_K})^{K} \leq \epsilon$ and 
\begin{align}
	\PP(  \cC_{\cS'}(\tau_1) | \tau_1 = \cT) \leq \PP( \tau_1 \leq \phi | \tau_1 = \cT) \\+ \PP(\cC_{\cS'}(\tau_1), \tau_1 > \phi | \tau_1 = \cT)
\end{align}
We have:
\begin{align}
	\PP(\tau_1 \leq \phi) &=  \sum_{t=1}^T a(1-a)^{t-1} \PP( \phi \ge t)  \\
			      &\leq a \EE(\phi)  \\
			      &\leq a \Phi\left( \sum_{i=1}^{K-1} b_i\right)
\end{align}
Using the previous lower bound on $\PP(\tau_1 = \cT)$:
\begin{align}
	\PP(\tau_1 \leq \phi | \tau_1 = \cT) &= \frac{\PP(\tau_1 \leq \phi , \tau_1 = \cT) }{  \PP(\tau_1 = \cT)} \\
					     &\leq \frac{a \Phi\left( \sum_{i=1}^{K-1} b_i\right) }{ \PP(\tau_1 = \cT)}  \\
					     &\leq  2 (K + 1) p_0 \alpha^{b_K} \Phi\left(\sum_{i=1}^{K-1} b_i\right)  \\
					     &\leq \epsilon
\end{align}
by definition of $(b_1,...,b_N)$. Finally, when $\tau_1 > \phi$, sources in $\cS'$ are in stationary state at time $\tau_1$, so that, using the strong Markov property
\begin{align}
\PP(\cC_{\cS'}(\tau_1), \tau_1 > \phi | \tau_1 = \cT) \\
\leq \PP(\cC_{\cS'}(\tau_1) |\tau_1 > \phi, \tau_1 = \cT)  = \psi(k,\alpha) \leq \Psi, 
\end{align}

\underline{Step 5: Drift bound} We get:
\begin{align}
\Delta_1(x) &\leq  \PP( \tau_1 = \cT ) (v(x_1 + 1) (2 \epsilon + \Psi) - v(x_1)) \\
&= \alpha^{-x_1} \PP( \tau_1 = \cT )  ((x_1 + 1)(2 \epsilon + \Psi) \alpha^{-1} - x_1)	\\
&=x_1( (2 \epsilon + \Psi) \alpha^{-1} - 1) \alpha^{-x_1} \PP( \tau_1 = \cT ) \\
&+ \alpha^{-x_1} \PP( \tau_1 = \cT )(2 \epsilon + \Psi)
\end{align}
Consider the first term:
\begin{align}
x_1( (2 \epsilon + \Psi) \alpha^{-1} - 1) \alpha^{-x_1} \PP( \tau_1 = \cT ) \leq -\epsilon p_0  x_1
\end{align}
where we used the fact that $(2 \epsilon + \Psi) \alpha^{-1} \leq 1 - \epsilon$ and $\PP( \tau_1 = \cT ) \ge a$.
Consider the second term: 
\begin{align}
\alpha^{-x_1} \PP( \tau_1 = \cT )(2 \epsilon + \Psi) \leq 3 p_0 T
\end{align}
since $\epsilon < 1$, $\Psi \le 1$ and $\PP( \tau_1 = \cT) \ge a$.
We have obtained the following drift bound:
\begin{align}
\Delta_1(x)  \le  -\epsilon p_0 x_1  + 3 p_0 T.
\end{align}
The above reasoning is valid for all $x$ with $x_1 \geq \bar{b}$. To the contrary, if $x_1 \le \bar{b}$ we have that $v(X_1(\cT + 1)) \le v(\bar{b}+ 1)$ almost surely. 
Hence for all $x$: 
\begin{align}
	\Delta_1(x) &\le (-\epsilon p_0 x_1 + 3 p_0 T)\indic\{ x_1 \geq \bar{b} \}  + v( \bar{b} + 1)) \indic\{ x_1 < \bar{b}\} \\
&\le -\epsilon p_0 x_1 \indic\{ x_1 \geq \bar{b} \} +  3 p_0 T + v( \bar{b} + 1).
\end{align}
Repeating the same reasoning for other components of $x$ and summing we obtain the drift bound:
\begin{align}
\Delta(x) &\leq  -\epsilon p_0 \sum_{i=1}^N x_i \indic\{ x_i \geq \bar{b} \} + N\left(3 p_0 T + v\left( \bar{b} + 1\right)\right),
\end{align}
where $\Delta(x) = \sum_{i \in [N]} \Delta_i(x)$. Therefore, there exists a finite set $F \subset \NN^N$ such that for all $x \not\in F$, $\Delta(x) \leq -1$ Recall that $\cT \leq T$ a.s., hence $\cT$ has finite expectation. Applying \cite[Theorem 2.1]{meyn2012}, presented below, ensures that $X$ is positive recurrent, which concludes the proof.

\begin{lem}[\cite{meyn2012}]\label{le:meyn}
	Consider $X$ a discrete time, irreducible Markov chain, $\{\cT_i , i \in \NN\} \subset \NN$ an increasing sequence of stopping times, $V$ a positive function, $C$ a finite set and a constant $b$ such that:
	\begin{align}
		\EE( V( X(\cT_{i+1})) | \cF_{\cT_i}) \leq V( X(\cT_{i})) - 1 + b \indic\{ X(\cT_{i}) \in C\}
	\end{align}
	and $\sup_i \EE( \cT_{i+1} - \cT_i|  \cF_{\cT_i} ) < \infty$. Then $X$ is positive recurrent.
\end{lem}

\subsection{Null Recurrence or Transience}

We now consider the case  $\alpha < \psi(N-1,\alpha)$ and prove that $X(t)$ cannot be positive recurrent. Let $\bar X(t) = (X_1(t),...,X_{N-1}(t))$. Define $\bar{\pi}$ the stationary distribution of $\bar X(t)$ in a system where source $N$ has been removed. Define set ${\cal Z} = \{0,1\}^{N-1} \setminus (0,...,0)$. For $x \in \NN$, we write $X(t) \sim \bar{\pi} \otimes x$ to denote $\bar X(t) \sim \bar{\pi}$ and $X_N(t) = x$. For distribution $\mu$, define $T(\mu) = \EE_{X(0) \sim \mu}(\zeta)$ with $\zeta = \min\{t: X_N(t) = 0\}$, the expectation of the first time $X_N(t) = 0$ if $X(0) \sim \mu$. $T(\mu)$ is a lower bound on the first expected hitting time of state $(0,...,0)$.

If $X(0) \sim \bar{\pi} \otimes x$, then either (i) $Z_{N}(0) = 0$ so $X(1) \sim \bar{\pi} \otimes x$ by stationarity or (ii) $Z(0) = (0,...,0,1)$ so that $X_N(1) = 0$ or (iii) $Z(0) = (z_1,...,z_{N-1},1)$ with $z \in \cZ$ so that $X(1) \sim \bar{\pi}_y \otimes (x + 1)$, where $\bar{\pi}_z$ is the distribution of the sum of $z$ and a random variable distributed as $\bar{\pi}$. Summing over the cases: 
\begin{align}
	T(\bar{\pi} \otimes x) &= 1 + (1 - p_0 \alpha^{x}) T( \bar{\pi} \otimes x) \\
&+  p_0 \alpha^{x} \sum_{z \in \cZ} \bar{F}(z) T(\bar{\pi}_z \otimes (x+1))
\end{align}
with $\bar{F}(z) = \PP(\bar Z(0) = z | \bar X(0) \sim  \bar{\pi})$, so that:
\begin{align}
T(\bar{\pi} \otimes x) = \frac{1}{p_0 \alpha^{x}} + \sum_{z \in \cZ} \bar{F}(z) T(\bar{\pi}_z \otimes (x+1)).
\end{align}
For any $z \in \cZ$, define $U_z$ the coupling time of $\bar X(t)$ in a system where source $N$ has been removed, with starting distribution $\bar{\pi}_z$. Define $W = \min \{t \geq 0: Z_N(t) = 1\}$ the first transmission time for source $N$. If $U_z\leq W$ then $X(U_z) \sim \bar{\pi} \otimes x$, hence:
\begin{align}
T(\bar{\pi}_z \otimes x) \geq \PP(U_z \leq W) T(\bar{\pi} \otimes x)
\end{align}
Since $U_z$ and $W$ are independent, and $W \sim$ Geometric$(p_0 \alpha^{x})$:
\begin{align}
	\PP(U_y \geq W) &= p_0 \alpha^{x} \sum_{t \geq 0} \PP(U_z \geq t) (1-p_0 \alpha^{x})^t.
\end{align}
As $\EE(U_z) < \infty$ we have $\lim_{t \to \infty} \PP(U_z \geq t) = 0$ and $\lim_{x \to \infty} \PP(U_z \geq W) = 0$ using the fact that if $(\alpha_t)_t$ is a positive convergence sequence, then $\lim_{t \to \infty}  \alpha_t = \lim_{\gamma \to 0} \sum_{t \ge 0} \gamma (1-\gamma)^{t} \alpha_t$. So there exists $x(\epsilon,z)$ such that for all $x \geq x(\epsilon,z)$ we have 
$
T(\bar{\pi}_y \otimes x) \geq (1 -\epsilon) T(\bar{\pi} \otimes x)
$
Since $\cZ$ is finite, setting $x(\epsilon) = \max_{z \in \cZ} x(\epsilon,z)$ the above holds for all $z \in \cZ$ and
\begin{align}
T(\bar{\pi} \otimes x) &\geq \frac{1 }{ p_0 \alpha^{x}} +  (1-\epsilon) T(\bar{\pi} \otimes (x+1))  \sum_{y \in \cZ} \bar{F}(y) \\
		       &= \frac{1}{p_0 \alpha^{x}} +  (1-\epsilon) \psi(N-1,\alpha) T(\bar{\pi} \otimes (x+1)).
\end{align}
By induction on $x$ we deduce that:
\begin{align}
T(\bar{\pi} \otimes x) \geq \frac{1 }{p_0 \alpha^{x}} \sum_{k \geq 0} \left(\frac{  (1-\epsilon) \psi(N-1,\alpha) }{\alpha}\right)^{k}   
\end{align}
For $\epsilon$ small enough $(1-\epsilon) \psi(N-1,\alpha) \geq \alpha$, hence $T(\bar{\pi} \otimes x) = \infty$ and the system is not positive recurrent.

\section{Numerical Experiments}
\label{sec:simu}

Finally, we compare numerical results, obtained with a discrete-time simulator, to the mean-field (MF) approximation, both in the stable regime, where the Markov chain is positive recurrent,  and in the meta-stable regimes. Observe that in the latter case, no stationary distribution exists and the mean-field approximation is inaccurate. For this reason, we consider the following  {\it finite-time} mean-field approximation  (FTMF): let $\tilde{\pi}_t(x)$ be the  probability that a given source has experienced $x$ successive collisions at time $t$. At  time $t=0$, we have $\tilde{\pi}_0=\delta_0$ (no collision); then, $\tilde{\pi}_{t+1}$ is deduced from $\tilde{\pi}_{t}$ by the transition probabilities from~\eqref{eq:mf_transition}:
\begin{align}
\tilde{\pi}_{t+1}(x) & = (1-p_0\alpha^x)\tilde{\pi}_t(x)+q(t)p_0\alpha^{x-1}\tilde{\pi}_t(x)\text{ for $x>0$,}
\end{align}
where  $q(t)$ is the probability that another source transmits at time $t$:
\begin{equation}
q(t) = 1 - \Big(1-p_0\sum_{x\geq 0}\alpha^x\tilde{\pi}_t(x)\Big)^{N-1}.
\end{equation}
FTMF is based on  the same decoupling assumption as MF, applied in finite time. It is expected that FTMF ultimately converges to MF and thus becomes inaccurate in the meta-stable regime. 

In order to facilitate the observation of meta-stability, we introduce the notion of time-scales $S \in \NN$, with  time-scale $S$ corresponding to time interval $2^S \le t < 2^{S+1}$. Each time-scale is preceded by a warm-up period of the same length (up to 1 slot).
The simulations presented here extend up to scale $S=30$, amounting to $2^{31}-1$ time slots in total. For reference, if one slot equals $50 \mu s$\cite{bianchi1996performance}, then $S=10$ corresponds to $50 ms$, $S=20$ to just under a minute, and $S=30$ to over half a day.

For each value of $S$, the performance metrics are averaged over the corresponding period. 
To illustrate several scenarios, we consider the following model parameters:
\begin{itemize}
    \item $N=4$ (few sources) and $N=1024$ (many sources);
    \item $\alpha=1/2$, so that the backoff time is doubled after each collision, mimicking the doubling of the contention window;
    \item $p_0=2/17$ and $p_0=2/3$: The former corresponds to the same average backoff time as with a  contention window of size ${\rm CWmin}=16$ (i.e., uniform distribution over $\{1, 2, \ldots, 16\}$); this is a typical setting in CSMA~\cite[Table I]{bianchi2}; the latter corresponds to a setting with high access probability, giving rise to meta-stability.
\end{itemize}
Our experiments are fully reproducible and the code is publicly available as a Python package (\footnote{available at \url{https://pypi.org/project/slotted-aloha-simulator/}}). 
\subsection{Performance}

We define goodput and occupancy as the mean values of $R(t)$ and $O(t)$ over a time-scale, where $R(t)$ (resp. $O(t)$) denotes a \emph{successful} transmission (resp. a  transmission) occurring at time $t$:
\begin{align}
	R(t) &= \mathds{1} \Big( \sum_{i \in [N]} Z_i(t) = 1 \Big) \\ 
O(t) &=\max_{i\in [N]}Z_i(t).
\end{align}

\begin{figure}[t!]
	\centering
	\vspace{-.2cm}
	\subfloat[$N=4$, $p_0=2/17$\label{fig:b4_125}]{
 \includegraphics[width=.5\textwidth]{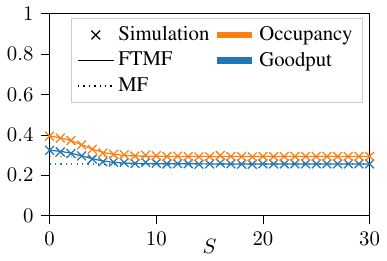}
 }
	\subfloat[$N=4$, $p_0=2/3$\label{fig:b4_500}]{
 \includegraphics[width=.5\textwidth]{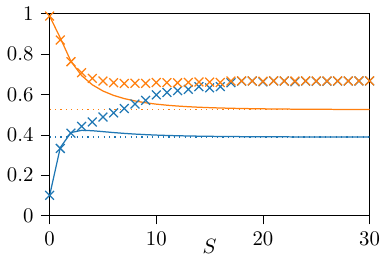}
 }\\
	\vspace{-.2cm}
	\subfloat[$N=1024$, $p_0=2/17$\label{fig:b1024_125}]{
 \includegraphics[width=.5\textwidth]{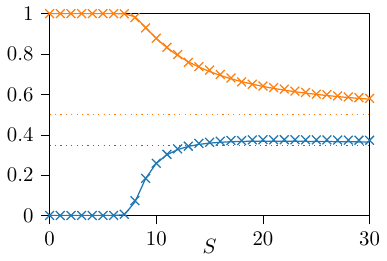}
 }
\subfloat[$N=1024$, $p_0=2/3$\label{fig:b1024_500}]{
\includegraphics[width=.5\textwidth]{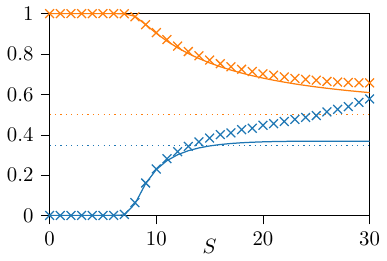}
}
	\caption{Occupancy and goodput on different time-scales $S$ ($\alpha=1/2$).\label{fig:basic}}
	\vspace{-.3cm}
\end{figure}

Figure~\ref{fig:basic} shows the evolution of goodput and occupancy for the considered parameters. In addition to the simulation results, we also display the values predicted by MF and FTMF, which are derived from the state distribution.

While FTMF is accurate for $p_0=2/17$ (Figures~\ref{fig:b4_125} and \ref{fig:b1024_125}), it is not for $p_0=2/3$ (Figures~\ref{fig:b4_500} and \ref{fig:b1024_500}), confirming that mean field  approximations do not work if the system is meta-stable. In particular, the goodput is above the predicted values. In Figure~\ref{fig:b4_500}, goodput matches occupancy at large scales: almost all transmissions are successful. The explanation, developed later, is that only one single source is able to emit at a given time, the others been forced into idleness.

For $N=4$, the metrics seem to converge in the scales considered, but this is not the case for $N=1024$, where occupancy still seems to evolve at $S=30$ (over half a day). This means that for large $N$, we may not be able to observe a steady state at reasonable scales, even if it exists.

\subsection{State Distribution}

\begin{figure}[!t]
	\centering\vspace{-.2cm}
	\subfloat[$N=4$, $p_0=2/17$\label{fig:d4_125}]{
 \includegraphics[width=.5\textwidth]{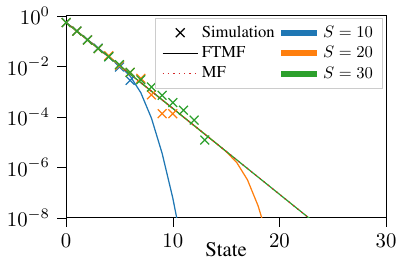}
 }
	\subfloat[$N=4$, $p_0=2/3$\label{fig:d4_500}]{
 \includegraphics[width=.5\textwidth]{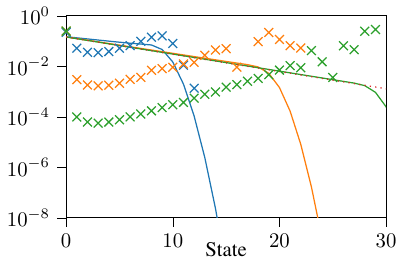}
 }\\
	\vspace{-.2cm}
	\subfloat[$N=1024$, $p_0=2/17$\label{fig:d1024_125}]{
 \includegraphics[width=.5\textwidth]{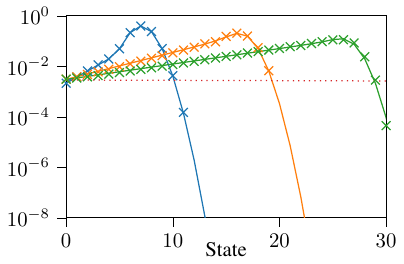}
 }
	\subfloat[$N=1024$, $p_0=2/3$\label{fig:d1024_500}]{
 \includegraphics[width=.5\textwidth]{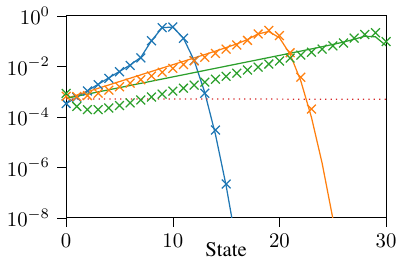}
 }\\
	\caption{Proportion of sources in each state on different time-scales ($\alpha=1/2$).\label{fig:distribution}}
\vspace{-.2cm}
\end{figure}

To have a better understanding of the results above, Figure~\ref{fig:distribution} shows the state distribution averaged over all sources, measured at time-scales $10$, $20$, and $30$.

For $N=4, p_0=2/17$ (Figure~\ref{fig:d4_125}), simulation fits the asymptotic geometric distribution predicted by MF (cf Section~\ref{sec:mean_field_geo}). In other settings, however, the simulations exhibit a ``positive slope'' incompatible with a stationary geometric distribution.

For $N=4, p_0=2/3$ (Figure~\ref{fig:d4_500}), state~$0$ has roughly probability 1/4, while the rest of the mass is concentrated in values that become larger and larger with $S$. This explains the performance observed in Figure~\ref{fig:b4_500}: one source is in state~$0$ and can emit at maximum speed while the others are idle due to their high states.

In all simulations, one can observe that the distribution collapses for large values. This is expected: the average time for a source $i$ to go from $x_i$ to $x_i+1$ is at least $2^{x_i}/p_0$. In particular, at scale $S$, it is unlikely to observe any state $x_i>S$: there is not enough time for these states to be reached.

From this observation, we draw the conclusion that at time-scale $S$, it is not possible to separate a meta-stable system from a stable system that has a non-negligible probability of having sources with a state greater than $S$. In such cases, only the transitory behavior of the system can be observed, and the stability condition becomes somehow irrelevant.

Consider for example the case $N=1024, p_0=2/17$ (Figure \ref{fig:d1024_125}): the simulations are not geometric, but FTMF fits simulations perfectly, which validates the mean-field model. We conjecture that the system is stable, but that its convergence can only be observed at impractical scales ($S=100$ or more, i.e. far greater than the age of the universe): only the transient behavior can be observed, where some sources get larger and larger states, like for a meta-stable system.

\subsection{Implicit Admission Control}

\begin{figure}[!t]
	\centering
	\subfloat[$p_0=2/17$\label{fig:l4}]{
 \includegraphics[width=.5\textwidth]{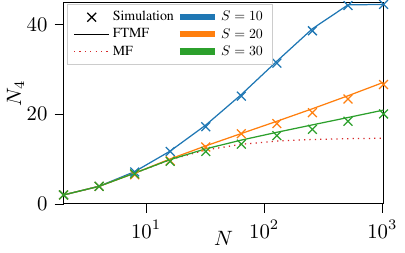}
 }
	\subfloat[$p_0=2/3$ (unstable)\label{fig:l1024}]{
 \includegraphics[width=.5\textwidth]{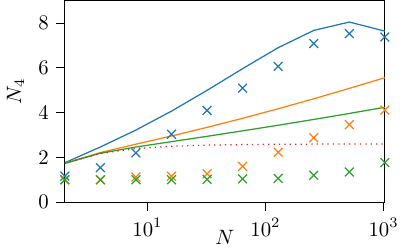}
 }
	\caption{Average number of sources $N_4$ with state $4$ or less, as a function of $N$ and time-scale $S$. \label{fig:live}}
	\vspace{-.3cm}
\end{figure}

Figure~\ref{fig:basic} shows that, except for small time-scales, the total goodput increases with the number of sources. However, if all sources are treated equally, the goodput of single source must go to $0$ as $N$ goes to infinity (the channel capacity is finite). Figures~\ref{fig:b4_500} and \ref{fig:d4_500} gave a first example where sources are not treated equally. This phenomenon, which we call \emph{implicit admission control}, should allow a slowly-evolving set of sources to have good access to the channel.

To study this, Figure~\ref{fig:live} shows, for different values of $S$ and $N$, the average number of sources $N_4$ with state equal to $4$ or less. $N_4$ gives an indication of the number of active sources with reasonable odds to transmit\footnote{The value $4$ is an arbitrary choice to separate active and idle sources. Other values yield qualitatively similar results.}. We can observe that $N_4$ grows with $N$, which demonstrates implicit admission control: even for a large value of $N$, a few sources stay active enough to achieve a throughput that does not go to $0$.

For $p_0=2/17$, $N_4$ grows up to $20$ for $N=1024, S=30$. If we assume that the mean-field approximation is accurate, which seems to be the case, $N_4$ should stay above $14$ even for large scales and values of $N$.

On the other hand, for $p_0=2/3$, we observe that $N_4$ rapidly goes to $1$ for large scales and small to medium values of $N$. The mean field approximation overestimates $N_4$, suggesting that the value is not likely to grow past $2$ for large scales and values of $N$. In other words, $p_0=2/3$ may be too aggressive with respect to admission control, as apart at small time-scales, it allows only one or two sources to be simultaneously active.

\subsection{Bounded States}

\begin{figure}[!t]
	\centering
	\subfloat[$N=4$, $p_0=2/17$\label{fig:state_4_0}]{
 \includegraphics[width=.5\textwidth]{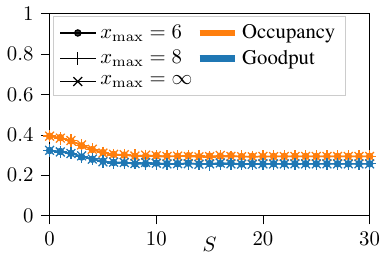}
 }
	\subfloat[$N=4$, $p_0=2/3$\label{fig:state_4_1}]{
 \includegraphics[width=.5\textwidth]{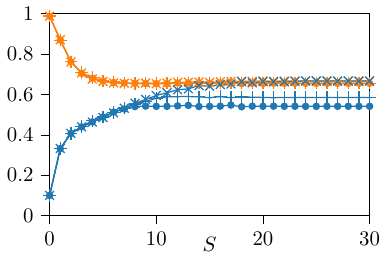}
 
 }\\
	\vspace{-.2cm}
	\subfloat[$N=1024$, $p_0=2/17$\label{fig:state_1024_0}]{
  \includegraphics[width=.5\textwidth]{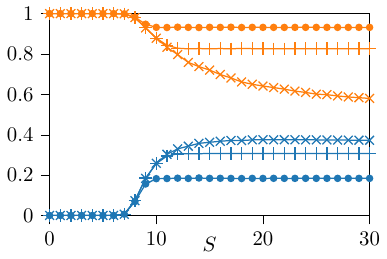}
 }
\subfloat[$N=1024$, $p_0=2/3$\label{fig:state_1024_1}]{
  \includegraphics[width=.5\textwidth]{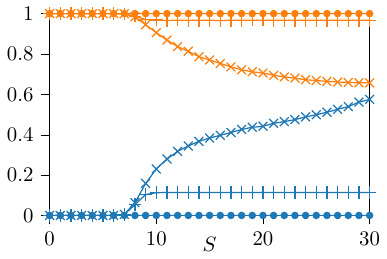}
}
	\caption{Impact of the number of states allowed ($\alpha=1/2$).\label{fig:states}}
	\vspace{-.3cm}
\end{figure}

Implicit admission control does not exist in traditional models, which limit the number of states. Indeed, bounding the number of states means that all sources have a minimal probability to try to transmit, so when $N$ grows, the probability of collision goes to $1$ and the total goodput goes to $0$.

In Figure~\ref{fig:states}, simulations show what happens if we limit the values of $x$ to $x_{\max}$ (if a source in state $x_{\max}$ is colliding, it stays in state $x_{\max}$). In addition to the unlimited model, we display the performance for $x_{\max}=6$ (for $p_0=2/17$, i.e. ${\rm CWmin=16}$ in the classic model, this translates to ${\rm CWmax}=1024$) and $x_{\max}=8$ (${\rm CWmax}=4096$).

When the system is very stable like in Figure~\ref{fig:state_4_0}, the limitation has no impact because it is unlikely to encounter many successive collisions. In other settings, we observe that the limitation reduces the goodput and increases the occupancy. In particular, for large $N$, when $p_0$ increases, the goodput increases when there is not bound but decreases otherwise. This type of interaction between the number of states and the goodput was already observed in~\cite{kumar2005new}.

Lastly, Figures~\ref{fig:state_1024_0} and \ref{fig:state_1024_1} suggest that, even if real systems may prefer to have a bounded $x$ (e.g. easier implementation, prevent sources from being blocked for too long), increasing even slightly the bound can have a large, positive impact on the goodput of the system when $N$ is large.

\subsection{Model Validation}
\label{sec:validation}

Our model relies on geometric distributions to decide emissions but standard protocols use uniform distributions over finite windows. 
In Figure~\ref{fig:shape}, we compare performance between the two. We can see that both models yield very similar results except for small $N$ and small time-scales. This difference can actually be explained: with parameters $p_0=2/17$ and  ${\rm CWmin}=16$), the expected delay between two transmissions at state $0$ is $17/2$ for both models. However, the probability to transmit at time $0$ differs: while it is $2/17$ in the geometric model, it is $1/16$ in the uniform one. As a consequence, for $N=4$, occupancy at time $0$ is $1-(15/17)^4\approx 0.39$ in the geometric model, but only $1-(15/16)^4\approx 0.23$ in the uniform one (goodput values can be computed the same way). As soon as the first uniform contention window is exhausted, what matters is the delay between two transmissions and the performance difference vanishes.
Experiments with other parameters (various $N$, $p_0=2/({\rm CWmin}+1)$, bounded states, etc.) show the same trend and validate our geometric approach.

\begin{figure}[!t]
	\centering\vspace{-.2cm}
	\subfloat[$N=4$ \label{fig:shape_4}]{
 \includegraphics[width=.5\textwidth]{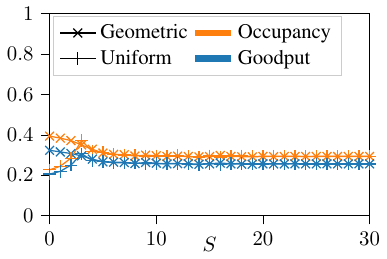}
 }
	\subfloat[$N=1024$ \label{fig:shape_1024}]{
 \includegraphics[width=.5\textwidth]{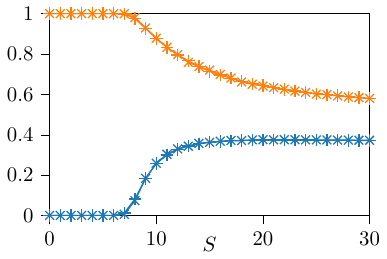}
 }
	\vspace{-.2cm}
	\caption{Impact of the shape of the contention window ($p_0=2/17$ and $\alpha=1/2$ for  geometric, ${\rm CWmin}=16$ for uniform). \label{fig:shape}}
	\vspace{-.2cm}
\end{figure}

\section{Conclusion}\label{sec:conclu}

We have shown that exponential backoff, in the realistic case where the protocol parameters are not dependent on the number of sources,  can have both stable and meta-stable behavior, and that meta-stability can be desirable, by creating a form of implicit admission control. Since in the meta-stable regime previously known analyses based on mean field approximation make incorrect predictions, we have provided both theoretical tools and extensive numerical experiments to understand this new phenomenon.

\bibliographystyle{plain} 
\bibliography{biblio}
\end{document}